\documentclass[11pt]{article}
\usepackage[a4paper, total={6in, 8in}]{geometry}
\usepackage{enumerate}
\usepackage{amsmath}
\usepackage{amssymb,latexsym}
\usepackage{amsthm}
\usepackage{color}
\usepackage{cancel}
\usepackage{graphicx}
\usepackage{cite}
\usepackage{changes}
\usepackage{lineno}
\usepackage{hyperref}
\usepackage{comment}
\usepackage{amscd}

\DeclareMathOperator{\C}{\mathcal{C}}

\DeclareMathOperator{\Aut}{Aut}
\DeclareMathOperator{\End}{End}

\DeclareMathOperator{\rk}{rk}

\newtheorem{theorem}{Theorem}[section]

\newtheorem{lemma}[theorem]{Lemma}
\newtheorem{corollary}[theorem]{Corollary}
\newtheorem{definition}[theorem]{Definition}
\newtheorem{proposition}[theorem]{Proposition}

\newtheorem{remark}[theorem]{Remark}

\newtheorem{openq}[theorem]{Open Question}
\newtheorem{question}[theorem]{Question}

\newcommand{\fqn}{\mathbb{F}_{q^n}}
\newcommand{\fqm}{\mathbb{F}_{q^m}}

\newcommand{\cL}{{\mathcal L}}

\newcommand{\F}{{\mathbb F}}

\newcommand{\GL}{\hbox{{\rm GL}}}

\newcommand{\fq}{{\mathbb F}_{q}}

\newcommand{\ev}{\mathrm{ev}}

\title{Divisible linear rank metric codes}
\author{Olga Polverino, Paolo Santonastaso, John Sheekey and Ferdinando Zullo}
\date{ }

\begin{document}
\maketitle

\begin{abstract}
A subspace of matrices over $\F_{q^e}^{m\times n}$ can be naturally embedded as a subspace of matrices in $\F_q^{em\times en}$ with the property that the rank of any of its matrix is a multiple of $e$. It is quite natural to ask whether or not all  subspaces of matrices with such a property arise from a subspace of matrices over a larger field. In this paper we explore this question, which corresponds to studying divisible codes in the rank metric. 
We determine some cases for which this question holds true, and describe counterexamples by constructing subspaces with this property which do not arise from a subspace of matrices over a larger field.
\end{abstract}


\section{Introduction}

A (linear) code is a subset (subspace) of a vector space equipped with a distance function. A code  $\C$ is said to be \textbf{$e$-divisible} if the distance between any two elements of $\C$ is divisible by $e$. A fundamental goal is the following.

\begin{center}
  {\it Construct or characterise all $e$-divisible codes with chosen parameters with respect to a given metric.}  
\end{center}

In this paper we address this question for the case of codes in the rank metric. Let $q$ be a power of a prime $p$ and $n,m$ be two positive integers.
The set $\mathbb{F}_q^{m\times n}$ can be equipped with the rank metric and any subset (subspace) of $\mathbb{F}_q^{m\times n}$ is called a \textbf{(linear) rank metric code}, we refer to \cite{gorla2018codes, bartz2022rank,polverino2020connections,sheekeysurvey} for more details on rank metric codes and their applications.
Consider a rank metric code $\C$ in $\F_{q^e}^{m\times n}$, for some positive integer $e$. Once we fix an $\mathbb{F}_{q^e}$-basis of $\F_{q^e}^m$ and an $\mathbb{F}_{q^e}$-basis of $\F_{q^e}^n$, $\C$ can be seen also as a set of $\mathbb{F}_{q^e}$-linear map from $\F_{q^e}^n$ in $\F_{q^e}^m$. 
Clearly, both $\F_{q^e}^m$ and $\F_{q^e}^n$ may be seen as $\F_q$-vector spaces and hence the maps of $\C$ also define $\F_q$-linear maps from $\F_{q^e}^n$ in $\F_{q^e}^m$, seen as $\fq$-vector spaces, i.e. the elements of $\C$ can be embedded in $\fq^{em \times en}$ as a rank metric code, let say $Em(\mathcal{C})$.
Clearly, the map $Em$ depends on the choice of the bases of the vector spaces involved.
This embedding preserves the rank of the elements of $\C$ up to multiplying by $e$, i.e. if $A \in \C$ has rank $j$, then the associated matrix via $Em$ will have rank $ej$.
In this way, all the elements of $Em(\mathcal{C})$ have rank divisible by $e$, and thus $\C$ is an \textbf{$e$-divisible rank metric code}.

Divisible codes in the Hamming metric have been deeply studied for their connection with interesting geometrical objects (see e.g. \cite{ball2007linear,landjev2019divisible,kurz2021generalization}), codes of nodal surfaces (see e.g. \cite{barth2007cusps}), subspace covering (see e.g. \cite{etzion2014covering,etzion2011q}), $q$-analogs of group divisible designs (see e.g. \cite{buratti2018q}) and very recently used for quantum computations \cite{hu2022divisible}.
Very few classification results are known, see e.g. \cite{betsumiya2012triply}.
Moreover, in the Hamming metric, $e$-divisible codes can be easily constructed using {\it repetition codes}, and classification results focus on determining whether or not all $e$-divisible codes are repetition codes, or arise from a short list of known constructions. In the Hamming metric, viewing codes over $\F_{q^e}$ as codes over $\fq$ does not lead to $e$-divisible codes (in fact, this process is known as {\it concatenation}, and is a topic of research in its own right), whereas repetition codes in the rank metric do not behave as well as those in the Hamming metric (in particular, in the rank metric some linearity properties are  not always preserved, see Remark \ref{rk:repetitionrank}); this is an important distinction between these two metrics.
For a complete survey on divisible codes in the Hamming metric we refer to \cite{kurz2021divisible}

The natural question that arises in the rank metric is whether or not every $e$-divisible rank metric code $\C$ in $\mathbb{F}_q^{m'\times n'}$ \textbf{arises from a rank metric code in} $\F_{q^e}^{m\times n}$ (or simply, \textbf{over $\F_{q^e}$}), i.e. if there exists a rank metric code $\C'$ in $\F_{q^e}^{m\times n}$ for which $\C=Em(\C')$.
As it is formulated, the above question has a negative answer. Indeed, the $\fq$-subspace of alternating matrices (that is, matrices $A$ such that $A^\top = -A$ and having zeros on the main diagonal) in $\fq^{m\times m}$ is $2$-divisible (as every alternating matrix has even rank), but if $m$ is odd it cannot arise from a rank metric code in $\F_{q^2}^{m/2\times m/2}$.
However, it seems that there are no counterexamples when considering codes with a \emph{greater linearity}. Consider a rank metric code $\C$ in $\fq^{m\times n}$, the set 
\[ L(\C)=\{ X \in \fq^{m\times m} \colon XA \in \mathcal{C}, \,\,\, \text{for any}\,\,\,  A \in \C \} \]
is called the \textbf{left idealiser} of $\C$, see \cite{liebhold2016automorphism,lunardon2018nuclei}.
In the case in which $(L(\C),+,\cdot)$ contains a subring $\mathcal{F}$ isomorphic to $\mathbb{F}_{q^m}$, where $+$ and $\cdot$ are the sum and the product of matrices, respectively, then we say that $\C$ is $\F_{q^m}$-\textbf{linear}. This is due to the fact that $\C$ turns out to be an $\mathcal{F}$-left vector space.
So, we can now state the question we are going to investigate in this paper.

\begin{question}\label{question}
Is it true that every $\F_{q^m}$-linear $e$-divisible rank metric code $\C$ in $\fq^{m\times n}$ arises from a rank metric code in $\F_{q^e}^{m/e\times n/e}$?
\end{question}

In this paper we answer positively to  Question \ref{question} in the case in which $n$ is a multiple of $m$ (see Section \ref{sec:3}).
The main tool regards the view of such a code as a subspace of linearized polynomials and the geometric version of these codes, together with the well-celebrated result on the number of directions determined by function in one and more variables (see Section \ref{sec:directions}). 
While in Section \ref{sec:negative} we describe examples of rank metric codes which will give a negative answer to Question \ref{question} under suitable assumptions on the parameters. This is done using the geometric approach of $q$-systems and using a tower of extension fields of a finite field.




\section{Preliminaries} 
\subsection{Directions of functions}\label{sec:directions}

Let $\mathrm{AG}(n,q)$ be the affine $n$-dimensional space and let $H_{\infty}$ be its hyperplane at infinity.
Consider a pointset $S$ in $\mathrm{AG}(n,q)$, then the set of directions determined by $S$ is 
\[ D=\{ \langle P,Q \rangle \cap H_{\infty} \colon P,Q \in S, \text{with}\,\, P\ne Q \}. \]

In the case in which $S$ coincides with the graph of a univariate function from $\fq$ to $\fq$, there is the following well celebrated result.

\begin{theorem}(see \cite{blokhuis1999number,ball2008graph}) \label{th:directions}
Let $f$ be a function from $\fq$ to $\fq$, $q=p^h$, and let $N$ be the number of directions determined by $f$.
Let $s=p^e$ be maximal such that any line with a direction determined by $f$ that is incident with a point of the graph of $f$ is incident with a multiple of $s$ points of the graph of $f$. 
Then one of the following holds:
\begin{itemize}
    \item $s=1$ and $(q+3)/2\leq N \leq q+1$;
    \item $e\mid h$, $q/s+1\leq N\leq (q-1)/(s-1)$;
    \item $s=q$ and $N=1$.
\end{itemize}
Moreover, if $s>2$ then the graph of $f$ is $\F_s$-linear.
\end{theorem}

An extension of Theorem \ref{th:directions} to higher dimensions is the following from \cite{storme2001linear} (see also \cite{ball2008graph}).

\begin{theorem} \label{th:dirmoreind} (see \cite[Theorem 9, Corollary 10]{storme2001linear})
Let $q=p^h$.
Let $S$ be a subset of $\mathrm{AG}(n,q)$, $|S|=q^{n-1}$ and let $D\subseteq H_{\infty}$ be the set of direction determined by $S$. Suppose that
\[|D| \leq \frac{q+3}2 q^{n-2}+q^{n-3}+\ldots+q^2+q.\]
Then for any line $l$ either
\begin{itemize}
    \item[(i)] $|S\cap l|=1$ (if and only if $l\cap H_{\infty} \notin D$), or
    \item[(ii)] $|S\cap l|\equiv 0 \pmod{p^{e_l}}$, for some $e_l\mid h$.
\end{itemize}
Moreover, $S$ is $\F_{q^e}$-linear, where $e$ is the greatest common divisor of the values $e_l$.
\end{theorem}

\subsection{Dual of $\fq$-subspaces}\label{sec:dual}

Let $V$ be an $\F_{q^m}$-vector space of dimension $k$ and let $\sigma \colon V \times V \rightarrow \F_{q^m}$ be a nondegenerate reflexive sesquilinear form on the $k$-dimensional $\F_{q^m}$-vector space $V$ and consider \[
\begin{array}{cccc}
    \sigma': & V \times V & \longrightarrow & \F_q  \\
     & (x,y) & \longmapsto & \mathrm{Tr}_{q^m/q} (\sigma(x,y)).
\end{array}
\] 
So, $\sigma'$ is a nondegenerate reflexive sesquilinear form on $V$ seen as an $\fq$-vector space of dimension $km$. Then we may consider $\perp$ and $\perp'$ as the orthogonal complement maps defined by $\sigma$ and $\sigma'$, respectively. For an $\F_q$-subspace $U$ of  $V$ of dimension $n$, the $\F_q$-subspace $U^{\perp'}$ is the \textbf{dual} (with resepct to $\sigma'$) of $U$, which has dimension $km-n$.

An important property that $\sigma'$ satisfies is that the dual of an $\F_{q^m}$-subspace $W$ of $V$ is an $\F_{q^m}$-subspace as well and $W^{\perp'}=W^\perp$. Moreover, the following result will be widely used in the paper.

\begin{proposition} (see \cite[Property 2.6]{polverino2010linear}) \label{prop:weightdual}
Let $U$ be an $\fq$-subspace of $V$ and $W$ be an $\F_{q^m}$-subspace of $V$.
Then 
\[ \dim_{\fq}(U^{\perp'}\cap W^{\perp})=\dim_{\fq}(U\cap W)+\dim_{\fq}(V)-\dim_{\fq}(U)-\dim_{\fq}(W). \]
\end{proposition}

\section{Rank metric codes}
Rank metric codes over finite fields can be represented in different setting. Let start by recalling the matrix framework for the rank metric. 

\subsection{Matrix codes}

On the space of $m \times n$ matrices over $\F_q$, we define the \textbf{rank distance} as the function 
\[
\mathrm{d} : \F_q^{m \times n} \times \F_q^{m \times n} \longrightarrow \mathbb{N}  
\]
defined by
\[
\mathrm{d}(X,Y) = \mathrm{rk}(X - Y),
\]
with $X,Y \in \F_q^{m \times n}$. 

We define the \textbf{rank weight} of an element $X \in \F_q^{m \times n}$ as 
$$\mathrm{w}(X):= \mathrm{rk}(X).$$
Clearly, $\mathrm{d}(X,Y)= \mathrm{w}(X-Y)$, for every $X,Y \in \F_q^{m \times n}$. \\
A \textbf{(linear matrix) rank metric code} $C$ is an $\F_q$-linear subspace of $\F_q^{m \times n}$ endowed with the rank distance.
The \textbf{minimum rank distance} of a rank metric code $C$ is defined as usual via $$\mathrm{d}(C)=\min\{\mathrm{w}(X): X \in C, X \neq 0\}.$$ 
We say that two linear rank metric codes $\C,\C' \subseteq \F_q^{m\times n}$ are \textbf{equivalent} if  there exist $X \in \GL(m, q)$, $Y \in \GL(n, q)$ and a field automorphism $\sigma$ of $\F_q$ such that
\[
\mathcal{C}'=X \mathcal{C}^{\sigma}Y =\{X A^{\sigma} Y : A \in \C\}.
\]

Therefore, we define a linear rank metric code $\C$ to be \textbf{$e$-divisible} if all the elements of $\C$ have rank a multiple of $e$. 

\begin{remark}\label{rk:repetitionrank}
Another way to construct $e$-divisible rank metric code is the following:
consider $\C'$ any linear rank metric code in $\mathbb{F}_q^{m\times m}$. Define the rank metric code $\C$  as the subset of $\mathbb{F}_q^{em\times em}$ whose elements are block diagonal matrices such that the main-diagonal blocks consist of the same element of $\C$, that is
\[ \C=\{\mathrm{diag}(C,\ldots,C) \colon C \in \C'\}. \]
The rank metric code $\C$ is clearly $e$-divisible. 
Note that if we consider $\C'\subseteq \F_q^{m\times m}$ where $\C'$ contains an element $\overline{C}$ of $\mathrm{GL}(m,q)$, then it is easy to see that
\begin{equation}\label{eq:leftidea} L(\C)=\{ \mathrm{diag}(A,\ldots,A) \colon A \in L(\C') \}. \end{equation}
Indeed, let $B$ be any matrix in $L(\C)$ and write it as a block matrix where the blocks are denoted by $B_{i,j}$ and are elements in $\F_q^{m\times m}$. Then
\[ B \cdot \mathrm{diag}(C,\ldots,C) = \left( 
\begin{matrix}
B_{1,1} C & B_{1,2} C & \ldots & B_{1,e} C\\
\vdots & \vdots & \vdots & \vdots \\
B_{e,1} C & B_{e,2} C & \ldots & B_{e,e} C\\
\end{matrix} \right)= \mathrm{diag}(C',\ldots,C'),  \]
for any $C \in \C'$ and some $C' \in \C'$.
Choose $C$ as $\overline{C}$ and denote by $O$ the zero matrix in $\C'$, then we have that 
\[ B_{i,j} \overline{C}=O, \]
for any $i,j \in \{1,\ldots,e\}$ with $i\ne j$, which implies that $B_{i,j}=O$ as $\overline{C}$ is invertible.
Moreover, we have 
\[ B_{i,i} \overline{C}=C', \]
for any $i$ and hence $B_{1,1}=\ldots=B_{e,e}$.
So, the matrices in $L(\C)$ are of the form $\mathrm{diag}(B,\ldots,B)$, and moreover $\mathrm{diag}(B,\ldots,B) \in L(\C)$ if and only if $BC \in \C'$ for every $C \in \C'$ and hence \eqref{eq:leftidea}.
Clearly, $L(\C)$ cannot contain a subring isomorphic to a proper extension of $\mathbb{F}_{q^{m}}$ since $L(\C')\simeq L(\C)$ and $L(\C')\subseteq \F_q^{m\times m}$. In particular, $\C$ cannot be $\mathbb{F}_{q^{me}}$-linear.
This gives us another important motivation to restrict our study to the case of rank metric codes with a larger linearity.
\end{remark}

\subsection{Vector codes}

We can also equivalently describe rank metric codes in terms of vectors.

The rank (weight) $w_q(v)$ of a vector $v=(v_1,\ldots,v_n) \in \F_{q^m}^n$ is the dimension of the vector space generated over $\F_q$ by its entries, i.e, $w_q(v)=\dim_{\fq} (\langle v_1,\ldots, v_n\rangle_{\fq})$. 

A \textbf{(linear vector) rank metric code} $\C $ is an $\F_{q^m}$-subspace of $\F_{q^m}^n$ endowed with the rank distance defined as
\[
d_q(x,y)=w_q(x-y),
\]
where $x, y \in \F_{q^m}^n$. 
If $q$ is clear from the context, we will just write $w,d$ in place of $w_q, d_q$.

Let $\C \subseteq \F_{q^m}^n$ be a rank metric code. We will write that $\C$ is an $[n,k,d]_{q^m/q}$ code (or $[n,k]_{q^m/q}$ code) if $k$ is the $\F_{q^m}$-dimension of $\C$ and $d$ is its minimum distance, that is 
\[
d=\min\{d(x,y) \colon x, y \in \C, x \neq y  \}.
\]




The matrix and vector framework for the rank metric described above are related in the following way. Let $\tilde{\Gamma}_q=(\gamma_1,\ldots,\gamma_m)$ be an ordered $\fq$-basis of $\F_{q^m}$. Given $x \in \F_{q^m}^n$, define the element $\Gamma_q(x) \in \F_q^{m \times n}$, where
$$x_{i} = \sum_{j=1}^m \Gamma_q (x)_{ji}\gamma_j, \qquad \mbox{ for each } i \in \{1,\ldots,n\}.$$
In other words, $\Gamma_q(x)$ is the matrix expansion (by columns) of the vector $x$ with respect to the $\fq$-basis $\tilde{\Gamma}_q$ of $\F_{q^m}$.
The map 
$$\Gamma_q: \F_{q^m}^n \rightarrow \F_q^{m\times n}$$
is an $\fq$-linear isometry between the metric spaces $(\F_{q^m}^n, d)$ and $(\F_q^{m \times n},\mathrm{d})$. 
Another important property of the map $\Gamma_q$ is that it keeps track of the $\F_{q^m}$-linearity.

\begin{proposition}
Let $\tilde{\Gamma}_q$ be an ordered $\fq$-basis of $\F_{q^m}$. If $\C$ is an $[n,k]_{q^m/q}$-code then $\C'=\Gamma_q(\C)$ is an $\F_{q^m}$-linear rank metric code in $\F_{q}^{m \times n}$. Conversely, let $\C$ be an $\F_{q^m}$-linear rank metric code in $\F_{q}^{m \times n}$ of dimension $mk$. Then there exists a code $\C'$ equivalent to $\C$ such that $\Gamma_q^{-1}(\C')$ is an $[n,k]_{q^m/q}$ code.
\end{proposition}
\begin{proof}
For an element $\alpha \in \F_{q^m}$, define
\[
\begin{array}{rrcl}
     \tau_{\alpha} : & \F_{q^m} & \longrightarrow & \F_{q^m}  \\
     & x & \longmapsto & \alpha x. 
\end{array}
\]
The set $\mathcal{A}$ of the matrices in $\F_q^{m \times m}$ associated with $\tau_{\alpha}$ with respect to the basis $\tilde{\Gamma}_q$, for every $\alpha \in \F_{q^m}$, forms a ring isomorphic to $\F_{q^m}$.
First, consider $\C'=\Gamma_q(\C)$ then we have to prove that $L(\C')$ contains a subring isomorphic to $\F_{q^m}$. 
We start by observing that 
\begin{equation}\label{eq:Gammaqlin}
  \Gamma_q(\alpha c) = A \Gamma_q(c),  
\end{equation}
for all $\alpha \in \F_{q^m}$ and $c \in \C$, where $A$ is matrix in $\F_q^{m \times m}$ associated with $\tau_{\alpha}$ with respect to the basis $\Gamma_q$. Now, since $\alpha c \in \C$ for all $\alpha \in \F_{q^m}$ and $c \in \C$, we get that $\mathcal{A} \subseteq L(\C')$ and the first part of the assertion is proved. Conversely, suppose that $\C$ is an $\F_{q^m}$-linear rank metric code in $\F_{q}^{m \times n}$ and let $\mathcal{G}$ be a subring of $L(\C')$ isomorphic to $\F_{q^m}$. Then $\mathcal{G} \setminus  \{0\}$ and $\mathcal{A} \setminus  \{0\}$ define two Singer cycles in $(\GL(m,q),\cdot)$, and since two Singer cycles are conjugate (see \cite[pag. 187]{huppert2013endliche} and \cite[Section 1.2.5 and Example 1.12]{hiss2011finite}), then there exists an invertible matrix $H \in \mathrm{GL}(m,q)$ such that 
\[H^{-1}  \mathcal{G} H=\mathcal{A}.\] 
Let $\C'=H^{-1} \C$. It can be easily verified that $L(\C')$ contains $\mathcal{A}$ and therefore $\Gamma_q^{-1}(\C')$ is an $\F_{q^m}$-subspace of $\F_{q^m}^n$. Indeed, let $\alpha\in \F_{q^m}$ and $c \in \Gamma_q^{-1}(\C')$, and denote by $A$ the matrix in $\F_q^{m \times m}$ associated with $\tau_{\alpha}$ with respect to the basis $\tilde{\Gamma}_q$.
In particular, there exists $X \in \C'$ such that $\Gamma_q(c)=X$ and, since $A \in L(\C')$, we also have that
\[A\Gamma_q(c)=AX \in \C'.\]
By Equation \eqref{eq:Gammaqlin}, this implies that $\alpha c\in \Gamma_q^{-1}(\C')$ and hence $\Gamma_q^{-1}(\C')$ is an $[n,k]_{q^m/q}$ code.
\end{proof}

For more details on this two settings, see e.g. \cite{gorla2018codes}. \\

As for the rank metric codes seen as subspaces of matrices, if we consider a linear rank metric code $\C$ in $(\F_{(q^e)^{m}}^n,d_{q^e})$ then $\Gamma_q^{-1}(Em(\Gamma_{q^e} (\C)))$ will be a rank metric code in $(\F_{q^{em}}^{en},d_q)$ whose weights are divisible by $e$. If a rank metric code $\C$ in $(\F_{q^{em}}^{en},d_q)$ can be obtained as $\Gamma_q^{-1}(Em(\Gamma_{q^e} (\C')))$ where $\C'$ is a linear rank metric code in $(\F_{(q^e)^{m}}^n,d_{q^e})$, then we say that $\C$ \emph{arises from a rank metric code in $(\F_{q^{em}}^{n},d_{q^e})$} (or simply, \emph{over $\F_{q^e}$}).

\subsection{Linearized polynomial codes}

In the case $n=m$, there is an alternative way to see the $\fq$-algebra of $n\times n$ matrices as the algebra of  \textbf{linearized polynomials} (or \textbf{$q$-polynomials}). Formally, a linearized polynomial is a polynomial of the shape
$$ f(x):=\sum_{i=0}^{t}f_i x^{q^i}, \quad f_i \in \fqn.$$
If $f$ is nonzero, the $q$-degree of $f$ will be the maximum $i$ such that $f_i \neq 0$.

The set of linearized polynomials forms an $\F_q$-algebra with the usual addition between polynomials and the composition, given by
$$ (f_ix^{q^i})\circ (g_jx^{q^j})=f_ig_j^{q^i}x^{q^{i+j}}, $$
on $q$-monomials, and then extended by associativity and distributivity and the multiplication by elements of $\F_q$. We denote this $\F_q$-algebra by $\cL_{n,q}[x]$.
$\mathcal{L}_{n,q}[x]$ can be considered modulo the two-sided ideal $(x^{q^n}-x)$ and so the elements of this factor algebra are represented by linearized polynomials of $q$-degree less than $n$, i.e.
$$ \tilde{\mathcal{L}}_{n,q}[x] := \left\{\sum_{i=0}^{n-1}f_i x^{q^i}, \quad f_i \in \fqn \right\}.$$
It is well-known that the following isomorphism as $\F_q$-algebra holds
\[ (\cL_{n,q}[x],+,\circ)\cong(\End_{\fq}(\fqn),+,\circ),
\] where the linearized polynomial $f(x)$ is identified with the endomorphism of $\fqn$
$$ \alpha \longmapsto \sum_{i=0}^{n-1}f_i \alpha^{q^i}.$$
Thanks to the above isomorphism, we immediately get that $\cL_{n,q}[x]$ is also isomorphic to the $\fq$-algebra $\fq^{n \times n}$, since $\F_{q^n}$ is an $n$-dimensional $\fq$-vector space. Here, we will speak of \emph{kernel} and \emph{rank} of a $q$-polynomial $f(x)$ meaning by this the kernel and rank of the corresponding endomorphism, denoted by $\ker(f(x))$ and $\rk(f(x))$, respectively. 

This naturally defines the rank metric in $\cL_{n,q}$. Indeed, in this context, a \textbf{linear rank metric code} $\mathcal{C}$ is an $\F_{q}$-subspace of $ \tilde{\mathcal{L}}_{n,q}[x]$ endowed with the rank metric.
Two linear rank metric codes $\C_1, \C_2 \subseteq \tilde{\mathcal{L}}_{n,q}$ are said to be \textbf{equivalent} if there exist two invertible linearized polynomials $g(x), h(x) \in \tilde{\mathcal{L}}_{n,q}$ and a field automorphism $\rho \in \Aut(\F_{q^n})$ such that
\[
\C_1=g(x) \circ \C_2^{\rho} \circ h(x) =\{g(x) \circ {f^{\rho}(x)} \circ h(x) : f(x) \in \C_2\},
\]
where $f^{\rho}(x)=\sum_{i=0}^{n-1}\rho(f_i)x^{q^i}$ if  $f(x)=\sum_{i=0}^{n-1}f_i x^{q^i}$.
In addition, as in the matrix framework we define a notion of idealizers in order to define a notion of linearity in this setting. Consider a linear rank metric code $\C$ in $\tilde{\mathcal{L}}_{n,q}$, the set 
\[ L(\C)=\{ g(x) \in  \tilde{\mathcal{L}}_{n,q} : g(x) \circ f(x) \in \C, \,\,\, \text{for any}\,\,\,   f(x) \in \C \} \]
is called the \textbf{left idealiser} of $\C$. We say that a linear rank metric code $\C$ is an $\F_{q^m}$-\textbf{linear rank metric code} if $(L(\C),+,\circ)$, where $+$ and $\circ$ are the sum and the composition of linearized polynomials, respectively, contains a subring isomorphic to 
\[\mathcal{F}_n=\{\alpha x : \alpha \in \F_{q^n}\} \simeq \mathbb{F}_{q^n}.\]
\begin{remark}
If $\C$ is an $\F_{q^n}$-linear rank metric code in $\tilde{\mathcal{L}}_{n,q}$ then it is equivalent to an $\F_{q^n}$-linear rank metric code $\C'$ that is also an $\F_{q^n}$-subspace of $\tilde{\mathcal{L}}_{n,q}$; similarly to what happen in the vector framework, that is there exists an invertible linearized polynomial $g(x) \in \tilde{\mathcal{L}}_{n,q}$ such that 
\[g^{-1}(x) \circ \mathcal{G}_n \circ g(x)=\mathcal{F}_n,\] 
where $\mathcal{G}_n$ is contained in $L(\C)$ and $\mathcal{G}_n$ is isomorphic to $\mathcal{F}_n$.
Let $\C'=g^{-1}(x) \circ \C$. It can be easily verified that $L(\C')$ contains $\mathcal{F}_n$. It follows that $\C'$ is an $\F_{q^n}$-linear subspace of $\tilde{\mathcal{L}}_{n,q}$.
\end{remark}

We now see the link with the vector model of rank metric codes.
Let $\mathcal{B}=(a_1,\ldots,a_{n})$ be an $\F_q$-basis of $\F_{q^n}$ (seen as an $\fq$-vector space).
The map
\[
\begin{array}{crcl}
        ev_{\mathcal{B}} \colon & \tilde{\mathcal{L}}_{n,q}[x] & \longrightarrow & \F_{q^n}^{n} \\
        &
        f(x) & \longmapsto & (f(a_1),\ldots,f(a_{n})) 
\end{array}
\]
is an $\F_{q^n}$-linear isomorphism which preserves the rank, i.e. $\rk(f(x))=w(ev_{\mathcal{B}}(f(x)))$.

\subsection{Rank metric codes and $q$-systems}\label{sec:qsystems}

Now, we recall the definition of equivalence between rank metric codes in $\F_{q^m}^n$. An $\F_{q^m}$-linear isometry $\phi$ of $\F_{q^m}^n$ is an $\F_{q^m}$-linear map of $\F_{q^m}^{n}$ that preserves the distance, i.e. $w(x)=w(\phi(x))$, for every $x \in  \F_{q^m}^{n}$, or equivalently $d(x,y)=d(\phi(x),\phi(y))$, for every $x,y \in  \F_{q^m}^{n}$.
It is known that the group of $\F_{q^m}$-linear isometries of $\F_{q^m}^n$ equipped with rank distance is generated by the (nonzero) scalar
multiplications of $\F_{q^m}$ and the linear group $\mathrm{GL}(n,\F_q)$, see e.g. \cite{berger2003isometries}. So we say that two rank metric codes $\C,\C' \subseteq \F_{q^m}^n$ are (linearly) equivalent if there exists an isometry $\phi$ such that $\phi(\C)=\C'$.
Clearly, when studying equivalence of $[n, k]_{q^m/q}$ codes the
action of $\F_{q^m}^*$ is trivial. This means that two $[n,k]_{q^m/q}$ codes $\C$ and $\C'$ are equivalent if and only if
there exists $A \in \mathrm{GL}(n,q)$ such that
$\C'=\C A=\{vA : v \in \C\}$. 
Most of the codes we will consider are \emph{non-degenerate}.

 \begin{definition}
An $[n,k]_{q^m/q}$ rank metric code $\C$ is said to be \textbf{non-degenerate} if the columns of any generator matrix of $\C$ are $\fq$-linearly independent. 
\end{definition}

The geometric counterpart of rank metric are the systems. 
 \begin{definition}
An $[n,k,d]_{q^m/q}$ \textbf{system} $U$ is an $\F_q$-subspace of $\F_{q^m}^k$ of dimension $n$, such that
$ \langle U \rangle_{\F_{q^m}}=\F_{q^m}^k$ and
$$ d=n-\max\left\{\dim_{\F_q}(U\cap H) \mid H \textnormal{ is an $\F_{q^m}$-hyperplane of }\F_{q^m}^k\right\}.$$
Moreover, two $[n,k,d]_{q^m/q}$ systems $U$ and $U'$ are \textbf{equivalent} if there exists an $\F_{q^m}$-isomorphism $\varphi\in\GL(k,\F_{q^m})$ such that
$$ \varphi(U) = U'.$$
We denote the set of equivalence classes of $[n,k,d]_{q^m/q}$ systems by $\mathfrak{U}[n,k,d]_{q^m/q}$.
\end{definition}

\begin{theorem}(see \cite{Randrianarisoa2020ageometric}) \label{th:connection}
Let $\C$ be a non-degenerate $[n,k,d]_{q^m/q}$ rank metric code and let $G$ be an its generator matrix.
Let $U \subseteq \F_{q^m}^k$ be the $\F_q$-span of the columns of $G$.
The rank weight of an element $x G \in \C$, with $x \in \F_{q^m}^k$ is
\begin{equation}\label{eq:relweight}
w(x G) = n - \dim_{\fq}(U \cap x^{\perp}),\end{equation}
where $x^{\perp}=\{y \in \F_{q^m}^k \colon x \cdot y=0\}.$ In particular,
\begin{equation} \label{eq:distancedesign}
d=n - \max\left\{ \dim_{\fq}(U \cap H)  \colon H\mbox{ is an } \F_{q^m}\mbox{-hyperplane of }\F_{q^m}^k  \right\}.
\end{equation}
\end{theorem}

Thanks the above Theorem, it has been proved that there is a one-to-one  correspondence  between  equivalence  classes of nondegenerate $[n,k,d]_{q^m/q}$ codes and equivalence classes of $[n,k,d]_{q^m/q}$ systems.
This correspondence can be formalized by the following two maps
\begin{align*}
    \Psi :  \mathfrak{C}[n,k,d]_{q^m/q} &\to\mathfrak{U}[n,k,d]_{q^m/q} \\
    \Phi : \mathfrak{U}[n,k,d]_{q^m/q} &\to \mathfrak{C}[n,k,d]_{q^m/q},
\end{align*}
that can be defined as follows.
Let $[\C]\in\mathfrak{C}[n,k,d]_{q^m/q}$ and $G$ be a generator matrix for $\C$. Then $\Psi([\C])$ is the equivalence class of $[n,k,d]_{q^m/q}$ systems $[U]$, where $U$ is the $\F_q$-span of the columns of $G$. In this case $U$ is also called a \textbf{system associate with} $\C$. Viceversa, given $[U]\in\mathfrak{U}[n,k,d]_{q^m/q}$. Define $G$ as the matrix whose columns are an $\F_q$-basis of $U$ and let $\C$ be the code generated by $G$. Then $\Phi([U])$ is the equivalence class of the $[n,k,d]_{q^m/q}$ codes $[\C]$. $\C$ is also called a \textbf{code associate with} $U$. $\Psi$ and $\Phi$ are well-posed and they are inverse of each other. See also \cite{alfarano2021linear}.

\begin{remark}
We note that the analogue of a repetition code for vector codes does not lead to divisible codes in the rank metric. For consider a vector $v\in \fqm^n$, and a vector $v'=v\oplus\cdots\oplus v\in \fqm^{ne}$. Then $w_q(v')=w_q(v)$, whereas $w_H(v')=e\cdot w_H(v)$, where $w_H$ denotes the Hamming weight.
\end{remark}

\section{Positive answers to Question \ref{question}}\label{sec:3}

In this section we answer positively to Question \ref{question} for rank metric codes in $\fq^{m\times m\ell}$. We will first analyze the case $\ell=1$, in which we will use the framework of linearized polynomials, the connection with $q$-systems and Theorem \ref{th:directions}. For the case $\ell>1$ we need to develop the connection with multivariate linearized polynomials and, as for the case $\ell=1$, then using the connection with $q$-systems together with Theorem \ref{th:dirmoreind} we get the desired result.

\subsection{Square case}

We start this subsection by applying Theorem \ref{th:directions} to rank metric codes which are two-dimensional. 


\begin{proposition}\label{prop:case2}
Let $\C=\langle f_1(x),f_2(x)\rangle_{\fqm}\subseteq \mathcal{L}_{m,q}[x]$ be a $2$-dimensional $\fqm$-linear rank metric code such that $f_1(x)$ is invertible and having second largest weight $m-e$.
Then $e \mid m$ and all the weights in $\C$ are divisible by $e$.
Moreover, $\C$ is equivalent to the code $\C'=\langle x,g(x)\rangle_{\fqm}$ with $g(x)=f_2\circ f_1^{-1}(x)$ which is $\F_{q^e}$-linear.
\end{proposition}
\begin{proof}
First consider $\C'=\C\circ f_1^{-1}=\langle x,g(x) \rangle_{\fqm}$, where $g(x)=f_2\circ f_1^{-1}(x)$.
So the code $\C'$ is equivalent to the code $\C$ and hence their weight spectrum coincide.
Let 
\[U=\{(x,g(x)) \colon x \in \fqm\} \subseteq \mathrm{AG}(2,q^m)\]
be the graph associated with the function $g$.
A line $\ell$ with slope $a$ meets the set $U$ in either zero or $q^t$ points, where $t=\dim_{\fq}(\{ z \in \mathbb{F}_{q^m} \colon g(z)=az \})$. Since the second largest weight of $\mathcal{C}$ (and hence of $\mathcal{C}'$) is $m-e$ then we know that there exists $a \in \mathbb{F}_{q^m}$ such that 
\[ \dim_{\mathbb{F}_q}(\ker(g(x)-ax))=e \]
and for every $b \in \F_{q^m}$ then $\mathrm{rk}(g(x)-bx)\leq m-e$ and hence
\[ \dim_{\fq}(\ker(g(x)-bx) )\geq e. \]
Therefore, $s=q^e$ is the largest $p$-power such that any line of $\mathrm{AG}(2,q^m)$ meeting $U$ in at least one point intersects $U$ in a multiple of $s$ points. 

We can now apply Theorem \ref{th:directions} obtaining that either $s=q^m$ and $g(x)=\lambda x$ for some $\lambda \in \fqm$ or that $\F_{q^e}$ is a proper subfield of $\fqm$ (and $e \mid m$). The former case cannot happen, otherwise the code $\C$ would not have dimension $2$.
So the latter case holds. By Theorem \ref{th:directions} if $q^e>2$ then $g$ is $\F_{q^e}$-linear, whereas if $q^e=2$ it follows that $e=1$ and $q=2$ and in this case we already know that $g$ is $\F_{q^e}$-linear. 
This implies that all the linearized polynomials in $\C'$ are $\F_{q^e}$-linear and so the rank of the elements in $\C'$ (and hence in $\C$) is divisible by $e$.
\end{proof}


We will use the above proposition as the base case for the general case.

\begin{theorem}\label{th:mxmqelinear}
Let $\C=\langle f_1(x),\ldots,f_k(x)\rangle_{\fqm}\subseteq \mathcal{L}_{m,q}[x]$ be a $k$-dimensional $\fqm$-linear rank metric code. Assume that $\C$ is $e$-divisible and that $\C$ contains an invertible linearized polynomial.
Then $e \mid m$ and $\C$ is equivalent to a code $\C'=\langle x,g_2(x),\ldots, g_k(x)\rangle_{\fqm}$ such that $g_2(x),\ldots,g_k(x)$ are $\F_{q^{e}}$-linear.
\end{theorem}
\begin{proof}
Without loss of generality, we may assume that $f_1(x)$ is an invertible element in $\C$.
As before, consider $\C'=\C\circ f_1^{-1}=\langle x,g_2(x),\ldots,g_k(x) \rangle_{\fqm}$, where $g_i(x)=f_i\circ f_1^{-1}(x)$ for every $i \in \{2,\ldots,k\}$.
In particular, the codes $\C$ and $\C'$ are equivalent and their weight spectra coincide.
By the assumptions, the weight of the elements in $\langle x, g_i(x) \rangle_{\fqm}$ are divisible by $e$ and hence the second maximum weight in the rank metric code $\langle x, g_i(x) \rangle_{\fqm}$ has the form $m-j_ie$, for some non-negative integer $j_i$. By Proposition \ref{prop:case2}, $g_i(x)$ is $\F_{q^{j_ie}}$-linear and hence $\F_{q^e}$-linear. The assertion is then proved.
\end{proof}


\begin{remark}
Note that in the proof of the last theorem we apply Proposition \ref{prop:case2} to the rank metric codes $\C_i=\langle x,g_i(x)\rangle_{\fqm}$ for any $i \in \{2,\ldots,k\}$ to obtain that the $g_i$'s are $\F_{q^e}$-linear.
We note here that the result holds true because all the $g_i$'s share the same copy of $\F_{q^e}$.
Indeed, recall that for a rank metric code $\mathcal{C}\subseteq \mathcal{L}_{m,q}[x]$ the \emph{centraliser} of $\mathcal{C}$ is defined as follows
\[ \mathrm{Cent}(\C)=\{h \in \mathcal{L}_{m,q}[x] \colon h \circ f=f \circ h,\,\,\forall f \in \C\} \]
and clearly if $\C_1$ and $\C_2$ are two rank metric codes with $\C_1 \subseteq \C_2$, then $\mathrm{Cent}(\C_1) \supseteq \mathrm{Cent}(\C_2)$.
Since the identity, that is $x$, is contained in all of these $\C_i$'s, then \[L(\C_i)\leq \C_i,\] 
so that
\[\mathrm{Cent}(\C_i)\leq \mathrm{Cent}(L(\C_i)).\] 
If $L(\C_i)$ contains $\mathcal{F}=\{\alpha x \colon \alpha \in \fqm\}\simeq\fqm$, then $\mathrm{Cent}(\C_i)\leq \mathrm{Cent}(\mathcal{F})=\mathcal{F}$ and hence it is uniquely determined.
\end{remark}

Hence we have the following.

\begin{corollary}\label{cor:matrixmmFqe}
Let $\C$ be an $\F_{q^m}$-linear rank metric code in $\fq^{m\times m}$ containing at least one invertible matrix. Then $\C$ is $e$-divisible if and only if it is equivalent to a code $\C'$ which arises from a rank metric code in $\F_{q^e}^{\frac{m}e\times \frac{m}e}$.
\end{corollary}
\begin{proof}
An $\F_{q^m}$-linear rank metric code $\C$ in $\fq^{m\times m}$ can be described as an $\F_{q^m}$-subspace $\mathcal{C}'$ of $\mathcal{L}_{m,q}[x]$ with the same weight spectrum.
Therefore, we can apply Theorem \ref{th:mxmqelinear} obtaining that the polynomials in $\mathcal{C}'$ are $\F_{q^e}$-linear, so that we may see $\mathcal{C}'$ as a subspace of $\mathbb{F}_{q^e}^{m/e\times m/e}$. In particular, we have that $\C=Em(\C')$.
\end{proof}

\subsection{Non-square case}

As already said in the previous section, we can describe any $\F_{q^m}$-linear rank metric code in $\F_{q^m}^m$ as an $\F_{q^m}$-subspace of $\mathcal{L}_{m,q}[x]$. This allowed us to prove Corollary \ref{cor:matrixmmFqe}. 

We will first extend the connection between linearized polynomials and linear rank metric codes, and then we will use this extension to generalize Corollary \ref{cor:matrixmmFqe}.

Let denote by $\tilde{\mathcal{L}}_{m,q}[x_1,\ldots,x_{\ell}]$  (or by $\tilde{\mathcal{L}}_{m,q}[\underline{x}]$) the $\fq$-vector space of linearized polynomials over $\F_{q^m}$ in the indeterminates $\underline{x}=(x_1,\ldots,x_{\ell})$ modulo the ideal $(x_1^{q^m}-x_1,\ldots,x_{\ell}^{q^m}-x_{\ell})$. As an $\fq$-vector space, $\tilde{\mathcal{L}}_{m,q}[\underline{x}]$ is $\fq$-linear isometric to the $\fq$-vector space of $\fq$-linear maps from $\F_{q^m}^\ell$ to $\F_{q^m}$ and to the $\fq$-vector space $\F_q^{m \times m\ell}$.

In the following result, we point out that studying $\F_{q^m}$-linear rank metric codes in $\F_{q^m}^{m\ell}$ is equivalent to study $\F_{q^m}$-subspaces of $\tilde{\mathcal{L}}_{m,q}[\underline{x}]$.

\begin{proposition}\label{prop:ev}
Let $\mathcal{B}=(a_1,\ldots,a_{m\ell})$ be an $\F_q$-basis of $\F_{q^m}^\ell$ (seen as an $\fq$-vector space).
The map
\begin{equation*}
    \begin{split}
        ev_{\mathcal{B}} \colon & \tilde{\mathcal{L}}_{m,q}[x_1,\ldots,x_{\ell}] \rightarrow \F_{q^m}^{m\ell} \\
        &
        f \mapsto (f(a_1),\ldots,f(a_{\ell m}))
    \end{split}
\end{equation*} 
is an $\F_{q^m}$-linear isomorphism which preserves the rank, i.e. $\mathrm{rk}(f)=\dim_{\fq}(\mathrm{Im}(f))=\dim_{\fq}(\langle f(a_1),\ldots,f(a_{\ell m}) \rangle_{\fq})$, for any $f \in \tilde{\mathcal{L}}_{m,q}[x_1,\ldots,x_{\ell}]$.
Moreover, if $\mathcal{C}=\langle f_1(\underline{x}),\ldots,f_k(\underline{x}) \rangle_{\F_{q^m}}$, then a generator matrix for $ev_{\mathcal{B}}(\C)$ is 
\[ G=\begin{pmatrix}
f_1(a_1) & f_1(a_2) & \ldots & f_1(a_{\ell m})\\
f_2(a_1) & f_2(a_2) & \ldots & f_2(a_{\ell m})\\
\vdots & & & \\
f_k(a_1) & f_k(a_2) & \ldots & f_k(a_{\ell m})
\end{pmatrix}.\]
\end{proposition}

Therefore, the map $ev_{\mathcal{B}}$ allows us to study rank metric codes as $\F_{q^m}$-subspaces of $\tilde{\mathcal{L}}_{m,q}[\underline{x}]$.
We are now interested in determining a \emph{canonical form} for the code, as the one we have in the square case when there is an invertible element. 
To this aim we will need the following auxiliary lemma. 

\begin{lemma}\label{lem:blockset}
Let $U$ be an $n$-dimensional $\fq$-subspace of an $\F_{q^m}$-vector space $V$ of dimension $r$, with $n<mr$. If $H$ is an $\F_{q^m}$-subspace of $V$ of dimension over $\F_q$ greater than $rm-n$, then $\dim_{\fq}(H\cap U)>0$.
Moreover, for any positive integer $l$ such that $lm\leq rm-n$, there exists an $\F_{q^m}$-subspace of $V$ such that $\dim_{\F_{q^m}}(S)=l$ and $S\cap U=\{0\}$.
\end{lemma}
\begin{proof}
Let $h=\dim_{\F_{q^m}}(H)$ and suppose that $mh >rm-n$, that is $\dim_{\fq}(H)+\dim_{\fq}(U)>rm$.
Since $\dim_{\fq}(U+H)\leq rm$, then we have $\dim_{\fq}(U\cap H)>0$.
For the second part of the assertion, let $l$ be the maximum positive integer (if it exists) such that $lm\leq rm-n$.
Let $i$ be the maximum non-negative integer such that there exists an $\F_{q^m}$-subspace $T$ of dimension $i$ in $V$ such that $T\cap U=\{0\}$ with $i<l$.
This implies that, by the first part, that $im\leq rm-n$ and all the $\F_{q^m}$-subspaces of dimension greater than $i$ meet $U$ in a non-trivial subspace.
The set of all the $(i+1)$-dimensional $\F_{q^m}$-subspaces through $T$ determines a partition of $U\setminus\{0\}$, from which we have
\[ q^n-1=|U|-1\geq (q-1)\frac{q^{m(r-i)}-1}{q^m-1}. \]
We can then derive that 
\[ n+m\geq rm-im+1\geq rm-(l-1)m+1, \]
that is $lm \geq rm-n+1$, a contradiction to the definition of $l$.
\end{proof}

We are now ready to provide a \emph{canonical form} for rank metric codes in $\tilde{\mathcal{L}}_{m,q}[\underline{x}]$.

\begin{theorem}\label{th:disjointkernel}
Let $\mathcal{C}=\langle f_1(\underline{x}),\ldots,f_k(\underline{x}) \rangle_{\F_{q^m}} \subset \tilde{\mathcal{L}}_{m,q}[x_1,\ldots,x_{\ell}]$ with $\dim_{\F_{q^m}}(\C)=k$ and $\ell <k$.
If $\dim_{\fq}\left( \bigcap_{i=1}^k \ker(f_i) \right)=0$, then there exist $\F_{q^m}$-linearly independent elements $g_1,\ldots,g_{\ell} \in \mathcal{C}$ such that 
\[ \bigcap_{i=1}^\ell \ker(g_i) =\{0\}. \]
\end{theorem}
\begin{proof}
Let consider $\varphi\colon a \in \F_{q^m}^\ell \mapsto (f_1(a),\ldots,f_k(a)) \in \F_{q^m}^k$, which clearly is an $\F_{q}$-linear map.
Since, $\dim_{\fq}\left( \bigcap_{i=1}^k \ker(f_i) \right)=0$ we have that $\ker(\varphi)=\{0\}$ and $\dim_{\F_{q}}(\mathrm{Im}(\varphi))=m \ell<mk$.
Because of the dimension of $\mathrm{Im}(\varphi)$, by Lemma \ref{lem:blockset}, there exists an $\F_{q^m}$-subspace $L$ of $\F_{q^m}^k$ of dimension $k-\ell$ such that $L\cap \mathrm{Im}(\varphi)=\{0\}$.
Suppose that $L$ corresponds to the set of the solutions of an homogeneous system associated with the matrix $A \in \F_{q^m}^{\ell \times k}$ of rank $\ell$.
It follows that the system in $\underline{x}$
\begin{equation}\label{eq:systcondpoldisj} A \begin{pmatrix}
f_1(\underline{x})\\
\vdots \\
f_k(\underline{x})
\end{pmatrix}
= O,
\end{equation}
where $O$ is the zero column in $\F_{q^m}^k$, has only the zero vector as solution, because of the fact that $\dim_{\fq}\left( \bigcap_{i=1}^k \ker(f_i) \right)=0$ and $L\cap \mathrm{Im}(\varphi)=\{0\}$.
Let
\[ A \begin{pmatrix}
f_1(\underline{x})\\
\vdots \\
f_k(\underline{x})
\end{pmatrix}
=  \begin{pmatrix}
g_1(\underline{x})\\
\vdots \\
g_\ell(\underline{x})
\end{pmatrix},
\]
then $g_1,\ldots,g_{\ell} \in \C$ and $\bigcap_{i=1}^\ell \ker(g_i)=\{0\}$, otherwise System \eqref{eq:systcondpoldisj} would admit a non zero solution.
The assertion then follows since the rank $A$ is $\ell$ and so the polynomials $g_1,\ldots,g_{\ell}$ are $\F_{q^m}$-linearly independent.
\end{proof}

\begin{remark}
Note that the condition $\dim_{\fq}\left( \bigcap_{i=1}^k \ker(f_i) \right)=0$ of the above theorem exactly corresponds to the non-degenerate condition on $\ev_{\mathcal{B}}(\C)$, for some $\fq$-basis $\mathcal{B}$ of $\F_{q^m}^\ell$.
Indeed, suppose that $\dim_{\fq}\left( \bigcap_{i=1}^k \ker(f_i) \right)>0$ and let $a \in \bigcap_{i=1}^k \ker(f_i) \setminus\{0\}$. Consider $\mathcal{B}$ any $\fq$-basis of $\F_{q^m}^\ell$ in which $a$ appears, namely $\mathcal{B}=(a,a_2,\ldots,a_{m\ell})$. Then when we consider the generator matrix associated with $ev_{\mathcal{B}}$ as in Proposition \ref{prop:ev}, then its first columns will be zero and hence $ev_{\mathcal{B}}(\C)$ is degenerate.
Conversely, assume that $ev_{\mathcal{B}}(\C)$ is degenerate and consider a generator matrix as in Proposition \ref{prop:ev}. Then one of the columns of $G$, without loss of generality we can suppose the first one, can be expressed as an $\fq$-linear combinations of the remaining column, namely
\[ \begin{pmatrix}
f_1(a_1)\\
\vdots\\
f_k(a_1)
\end{pmatrix}=\alpha_1 \begin{pmatrix}
f_1(a_2)\\
\vdots\\
f_k(a_2)
\end{pmatrix} +\ldots +\alpha_{\ell m-1} \begin{pmatrix}
f_1(a_{\ell m})\\
\vdots\\
f_k(a_{\ell m})
\end{pmatrix},\]
where $\alpha_i \in \fq$. Since the polynomials are linear over $\fq$ we obtain that $a_1-\alpha_1 a_2-\ldots-\alpha_{\ell m-1}a_{\ell m} \in \ker(f_i)\setminus\{0\}$ for any $i$, that is $\dim_{\fq}\left( \bigcap_{i=1}^k \ker(f_i) \right)>0$.
\end{remark}

\begin{remark}
The notion of equivalence of rank metric codes can be also read in the model of linearized polynomials. Indeed, the linear equivalence of rank metric codes $\F_{q}^{m\times \ell m}$ is given by the natural action of the following group on $\F_{q}^{m\times \ell m}$
\[ \{ (A,B) \colon A \in \mathrm{GL}( m,q),B\in \mathrm{GL}(\ell m,q) \}. \]
In terms of linearized polynomials, equivalent rank metric codes can be defined via the natural action of the following group on $\tilde{\mathcal{L}}_{m,q}[\underline{x}]$ 
\[ \{ (g,f) \colon g \in \tilde{\mathcal{L}}_{m,q}[x], f=(f_1,\ldots,f_{\ell})\in(\tilde{\mathcal{L}}_{m,q} [\underline{x}])^\ell \text{ are invertible} \}. \]
\end{remark}

\begin{corollary}
Let $\mathcal{C}=\langle f_1(\underline{x}),\ldots,f_k(\underline{x}) \rangle_{\F_{q^m}} \subset \tilde{\mathcal{L}}_{m,q}[x_1,\ldots,x_{\ell}]$ with $\dim_{\F_{q^m}}(\C)=k$, $\ell <k$ and 
\[ \bigcap_{i=1}^k \ker(f_i) =\{0\}. \]
Then $\C$ is equivalent to a rank metric code of the form
\[ \C'=\langle x_1,\ldots,x_{\ell},g_{\ell+1}(\underline{x}),\ldots,g_{k}(\underline{x}) \rangle_{\F_{q^m}}, \]
for some $g_{\ell+1},\ldots,g_{k} \in \tilde{\mathcal{L}}_{m,q}[x_1,\ldots,x_{\ell}]$.
\end{corollary}
\begin{proof}
By Theorem \ref{th:disjointkernel}, we may assume that 
\[ \bigcap_{i=1}^\ell \ker(f_i) =\{0\}. \]
Now, consider the following map $\varphi=(f_1,\ldots,f_{\ell})$, that is
\[ \varphi \colon a \in \F_{q^m}^\ell \mapsto (f_1(a),\ldots,f_{\ell}(a))\in \F_{q^m}^\ell. \]
This is an $\F_{q}$-linear isomorphism due to the fact that $\bigcap_{i=1}^\ell \ker(f_i) =\{0\}$.
Consequently, 
\[ (f_1\circ \varphi^{-1}(\underline{x}),\ldots,f_\ell\circ \varphi^{-1}(\underline{x}))=(x_1,\ldots,x_{\ell}). \]
Therefore, the rank metric code \[\C'=\C\circ {\varphi^{-1}} =\langle x_1,\ldots,x_{\ell},g_{\ell+1}(\underline{x}),\ldots,g_{k}(\underline{x}) \rangle_{\F_{q^m}},  \]
where $g_i=f_i\circ \varphi^{-1}$ and the assertion is proved.
\end{proof}

We can now extend Theorem \ref{th:mxmqelinear} for codes in $\mathcal{L}_{m,q}[\underline{x}]$ when $\ell >1$.

\begin{theorem}
Suppose $q>2$. Let $\C=\langle f_1(\underline{x}),\ldots,f_k(\underline{x})\rangle_{\fqm}\subseteq \mathcal{L}_{m,q}[x_1,\ldots,x_{\ell}]$ be a $k$-dimensional $\fqm$-linear rank metric code such that $\dim_{\fq}\left( \bigcap_{i=1}^k \ker(f_i) \right)=0$ and $\C$ is $e$-divisible.
Then $\C$ is equivalent to a code $\C'=\langle x_1,\ldots,x_{\ell},g_{\ell+1}(\underline{x}),\ldots, g_k(\underline{x})\rangle_{\fqm}$ such that $g_{\ell+1}(\underline{x}),\ldots,g_k(\underline{x})$ are $\F_{q^e}$-linear.
\end{theorem}
\begin{proof}
Let write $\C=\langle x_1,\ldots,x_{\ell},g_{\ell+1}(\underline{x}),\ldots, g_k(\underline{x})\rangle_{\fqm}$. First note that $\rk(x_1)=m$ and since $\C$ is $e$-divisible then $e \mid m$.
Let 
\[U=\{(x_1,\ldots,x_{\ell},g_{\ell+1}(\underline{x})) : x_i \in \F_{q^m} \} \subseteq  \F_{q^m}^{\ell+1} \]
be the system associated with $\C'=\langle x_1,\ldots,x_{\ell},g_{\ell+1}(\underline{x}) \rangle_{\F_{q^m}}$ and note that $\dim_{\fq}(U)=\ell m$. 
Since the code $\C$ is $e$-divisible, then $\C'$ is $e$-divisible. This implies that $e$ divides $\dim_{\fq}(U\cap H)$, for any hyperplane $H$ of $\F_{q^m}^{\ell+1}$ by \eqref{eq:relweight} since $e\mid m$.
Consider now $U^{\perp'}$, via the duality described in Section \ref{sec:dual}.
By Proposition \ref{prop:weightdual}, we have
\[ \dim_{\fq}(U^{\perp'}\cap H^{\perp})=\dim_{\fq}(U\cap H)-(\ell-1)m, \]
that is $e \mid \dim_{\fq}(U^{\perp'}\cap \langle w \rangle_{\F_{q^m}})$, for any $w \in \F_{q^m}^{\ell+1} \setminus\{0\}$.
Now, let $H$ be any hyperplane of $\F_{q^m}^{\ell+1}$ and let $i=\dim_{\fq}(U^{\perp'}\cap H)$. 
For any $\alpha \in \mathbb{N}$ denote
\[ N_{\alpha}=|\{\langle w\rangle_{\mathbb{F}_{q^m}} \colon w \in \mathbb{F}_{q^m}^{\ell+1}\setminus\{0\}\,\,\text{and}\,\, \dim_{\fq}((U^{\perp'}\cap H)\cap \langle w\rangle_{\mathbb{F}_{q^m}})=\alpha \}|. \]
Since the one-dimensional $\F_{q^m}$-subspaces of $\mathbb{F}_{q^m}^{\ell+1}$ give a partition of $\mathbb{F}_{q^m}^{\ell+1}\setminus\{0\}$ and $e \mid \dim_{\fq}(U^{\perp'}\cap \langle w \rangle_{\F_{q^m}})$,for any $w \in \F_{q^m}^{\ell+1} \setminus\{0\}$, we have that
\[ q^i-1=\sum_j N_{je}(q^{je}-1), \]
which implies that $q^e-1 \mid q^i-1$ and hence $e \mid i$. Therefore, $e \mid \dim_{\fq}(U^{\perp'}\cap H)$, for any hyperplane $H$ of $\F_{q^m}^{\ell+1}$.
Again, using Proposition \ref{prop:weightdual} we obtain that $e \mid \dim_{\fq}(U\cap \langle w\rangle_{\F_{q^m}})$, for any $w \in \F_{q^m}^{\ell+1}\setminus\{0\}$.
Now, consider
\[ S=\{ (x_1,\ldots,x_{\ell},g_{\ell+1}(\underline{x}),1) \colon x_1,\ldots,x_{\ell}\in \F_{q^m}\} \subset \mathrm{AG}(\ell+1,q^m) \]
and 
\[ D=\{ \langle (x_1,\ldots,x_{\ell},g_{\ell+1}(\underline{x}),0) \rangle_{\F_{q^m}} \colon x_1,\ldots,x_{\ell}\in \F_{q^m}\} \subset H_{\infty}, \]
where $D$ is the set of the directions determined by $S$. 
Note that since
\[ \overline{U}=\{  (x_1,\ldots,x_{\ell},g_{\ell+1}(\underline{x}),0) \colon x_1,\ldots,x_{\ell}\in \F_{q^m}\} \subset \F_{q^m}^{\ell+1} \]
is an $\fq$-subspace of dimension $m\ell$ and the one-dimensional $\F_{q^m}$-subspaces of $\mathbb{F}_{q^m}^{\ell+1}$ give a partition of $\mathbb{F}_{q^m}^{\ell+1}\setminus\{0\}$, it follows that 
\[ |D|\leq \frac{q^{m\ell}-1}{q-1}. \]
The above inequality and the fact that every element of $D$ meets $\overline{U}$ in an $\fq$-subspace of dimension multiple of $e$ allow us to apply Theorem \ref{th:dirmoreind} when $q>2$, obtaining that $S$ is $\F_{q^e}$-linear, from which it follows that the function $g_{\ell+1}$ is $\F_{q^e}$-linear.
We can argue as before, by replacing $g_{\ell+1}$ with the functions $g_{\ell+2},\ldots,g_k$ obtaining that they are $\F_{q^e}$-linear, that is the assertion.
\end{proof}

Similarly to what we did for Corollary \ref{cor:matrixmmFqe}, we have the following result.

\begin{corollary}\label{cor:matrixmlmFqe}
Let $\C$ be an $\F_{q^m}$-linear rank metric code in $\fq^{m\times \ell m}$ for which the intersection of the kernels of the matrices in $\C$ is trivial. Then $\C$ is $e$-divisible if and only if it arises from a rank metric code in $\F_{q^e}^{\frac{m}e\times \frac{\ell m}e}$.
\end{corollary}

\section{Negative answers to Question \ref{question}}\label{sec:negative}

In this section we provide some negative answers to Question \ref{question} by constructing examples of $e$-divisible rank metric codes which do not arise from a rank metric code in $(\F_{q^{em}}^n,d_{q^e})$, by making use of the geometric correspondence described in Section \ref{sec:qsystems}.

\begin{proposition}\label{prop:no2dim}
Consider the following field extension $\mathbb{F}_{q^{t\ell}}/\mathbb{F}_q$ and let $e$ a positive integer less than $t$.
Let $S_1$ be an $\fq$-subspace of $\mathbb{F}_{q^t}$ with dimension $e$ such that $1 \in S_1$ and $S_2$ be an $\F_{q^t}$-subspace of $\mathbb{F}_{q^{t\ell}}$ and $\dim_{\fq}(S_2)=g e$.
Then $U=S_1 \times S_2 \subseteq \mathbb{F}_{q^{t\ell}}^2$ is an $\fq$-subspace of dimension $e+ge$ with the property that
\begin{equation}\label{eq:weight2dim} 
\dim_{\fq}(U \cap \langle w\rangle_{\mathbb{F}_{q^{t\ell}}})\in \{ 0,e,ge \}, 
\end{equation}
for any $w \in \mathbb{F}_{q^{t\ell}}^2$.
\end{proposition}
\begin{proof}
The first part of the statement is trivial and we need to prove \eqref{eq:weight2dim}.
To this aim, let $w=(w_1,w_2) \in U$, with $w_1 \neq 0$, and let $\{\lambda_1,\ldots,\lambda_e\}$ be an $\fq$-basis of $S_1$. Since $S_2$ is an $\F_{q^t}$-subspace and since $S_1$ is contained in $\F_{q^t}$, then $(\lambda_1, \lambda_1 w_2/w_1),\ldots,(\lambda_e, \lambda_e w_2/w_1)$ are $e$ $\F_q$-linearly independent vectors of $U$ and so $\dim_{\fq}(U \cap \langle w \rangle_{\F_{q^{t \ell}}}) \geq e$. Now, since $\dim_{\fq}(U \cap \langle (0,1) \rangle_{\F_{q^{t\ell}}}) =ge$ and $(U \cap \langle w \rangle_{\F_{q^{t\ell}}}) \oplus (U \cap \langle (0,1) \rangle_{\F_{q^{t \ell}}})$ is an $\F_q$-subspace of $U$ we get that $\dim_{\fq}(U \cap \langle w \rangle_{\F_{q^{t\ell}}}) \leq e$. Therefore,  $\dim_{\fq}(U \cap \langle w \rangle_{\F_{q^{t \ell}}})=e$ and so
\[ \dim_{\fq}(U \cap \langle w \rangle_{\F_{q^{t \ell}}})\in \{0,e,ge\}, \]
for every $w \in \F_{q^{t\ell}}^2$. 
\end{proof}

Now, viewing $U$ of the above result as a $q$-system we obtain a class of $2$-dimensional linear divisible rank metric codes which cannot arise from a rank metric code in $(\F_{q^{em}}^n,d_{q^e})$.

\begin{corollary}
Let $t,g,e,\ell$ positive integers such that $e<t$, $ge \leq \ell t$ and $e\nmid t\ell$.
There exists a non-degenerate $e$-divisible linear rank metric of dimension two in $\F_{q^{t\ell}}^{e+ge}$ which cannot arise from a rank metric code over $\mathbb{F}_{q^e}$.
\end{corollary}
\begin{proof}
Let $U$ be as in Proposition \ref{prop:no2dim} and note that $U$ is a $[e+ge,2]$-system as $U$ contains the vectors $(1,0)$ and $(0,1)$. Let $\mathcal{C}$ be a code associated with $U$. By Theorem \ref{th:connection} and \eqref{eq:weight2dim}, it follows that the codewords of $\C$ have weight $e+ge, e+ge-e, e+ge-ge$, that is $e+ge, ge, e$. Therefore, the code $\C$ is a  non-degenerate $e$-divisible linear rank metric of dimension two in $\F_{q^{t\ell}}^{e+ge}$ and since $e$ does not divide $t\ell$, it cannot arise from a rank metric over $\F_{q^e}$.
\end{proof}

\section{Conclusions and open problems}

In this paper we answer to Question \ref{question}, providing a positive answer in the case in which the number of columns of the code is a multiple of $m$. 
We also find examples for which Question \ref{question} has a negative answer.

We conclude the paper by listing some open problems/questions that can be further explored.

\begin{openq}
Are there other conditions on the parameters involved in Question \ref{question} for which it has a positive answer?
\end{openq}

In the examples constructed in Section \ref{sec:negative}, in order to prove that an $e$-divisible rank metric code cannot arise from a larger field we use the fact that $e \nmid m$. So, the following question arises.

\begin{openq}
Are there any $e$-divisible linear rank metric in $\F_{q^m}^n$ for which $e\mid m$ but it does not arise from $\F_{q^e}$?
\end{openq}

The last question regards the linearity of $q$-systems and the embedabbility of the code.

\begin{openq}
How is the linearity of the $q$-system related to the matrix space in which the code can be embedded?
\end{openq}

 \section*{Acknowledgements}
The research of Olga Polverino, Paolo Santonastaso and Ferdinando Zullo was supported by the project ``VALERE: VAnviteLli pEr la RicErca" of the University of Campania ``Luigi Vanvitelli''. This  work  was  supported  by  the  “National  Group  for  Algebraic  and  Geometric  Structures, and their Applications” (GNSAGA – INDAM).

John Sheekey would like to thank the University of Campania ``Luigi Vanvitelli'' and his co-authors for their hospitality during the preparation of this paper.

\bibliographystyle{abbrv}
\bibliography{biblio}

John Shekeey\\
School of Mathematics and Statistics,\\ 
University College Dublin, Dublin, Ireland\\
{{\em john.sheekey@ucd.ie}}

\bigskip

Olga Polverino, Paolo Santonastaso and Ferdinando Zullo\\
Dipartimento di Matematica e Fisica,\\
Universit\`a degli Studi della Campania ``Luigi Vanvitelli'',\\
I--\,81100 Caserta, Italy\\
{{\em \{olga.polverino,paolo.santonastaso,ferdinando.zullo\}@unicampania.it}}

\end{document}